\documentclass[aps,prl,reprint,superscriptaddress]{revtex4-2}
\usepackage {CJKutf8}
\usepackage{amsmath}
\usepackage{pgfplots}
\usepackage{pgfplotstable}
\usetikzlibrary{intersections, patterns, calc, angles, quotes, fillbetween}

\newcounter{myequation}
\makeatletter
\@addtoreset{equation}{myequation}
\makeatother

\usepackage{amssymb}
\usepackage{amsthm}

\usepackage{enumitem}

\usepackage{todonotes}

\usepackage{float} 

\usepackage{geometry}
\theoremstyle{plain}%
\newtheorem{theorem}{Theorem}
\newtheorem{proposition}[theorem]{Proposition}%
\newtheorem{lemma}[theorem]{Lemma}

\theoremstyle{remark}%

\theoremstyle{definition}%
\newtheorem{definition}[theorem]{Definition}%

\DeclareMathOperator{\dist}{dist}
\DeclareMathOperator{\sgn}{sgn}

\newcommand{\di}{\partial}

\DeclareMathOperator{\Tr}{Tr}

\newcommand{\br}[1]{\left\langle#1\right\rangle}

\newcommand{\al}{\alpha}

\newcommand{\Cb}{\mathbb{C}}

\newcommand{\Zb}{\mathbb{Z}}
\newcommand{\Rb}{\mathbb{R}}

\newcommand{\one}{\ensuremath{\mathbf{1}}}

\newcommand{\Om}{\Omega}
\newcommand{\Ga}{\Gamma}

\newcommand{\z}{\zeta}

\newcommand{\abs}[1]{\ensuremath{\left\lvert#1\right\rvert}}
\newcommand{\norm}[1]{\ensuremath{\left\lVert#1\right\rVert}}

\newcommand{\Set}[1]{\left\{#1\right\}}

\newcommand{\md}[6]{\ensuremath{
		\ifinner
		\tfrac{\partial{^{#2}}#1}{\partial{#3^{#4}}\partial{#5^{#6}}}
		\else
		\tfrac{\partial{^{#2}}#1}{\partial{#3^{#4}}\partial{#5^{#6}}}
		\fi
}}
\newcommand{\del}[1]{\Bigl(#1\Bigr)}
\newcommand{\thmref}[1]{Theorem~\ref{#1}}

\newcommand{\lemref}[1]{Lemma~\ref{#1}}

\newcommand{\figref}[1]{Figure~\ref{#1}}

\usepackage[normalem]{ulem}
\usepackage{soul}
\setstcolor{blue}

\definecolor{green}{rgb}{0.0, 0.5, 0.5}

\definecolor{lgray}{gray}{0.9}
\definecolor{llgray}{gray}{0.95}
\definecolor{lllgray}{gray}{0.975}

\newcommand{\cS}{\mathcal{S}}

\newcommand{\nc}{\newcommand}

\nc{\h}{\delta}
\nc{\G}{\Gamma}
\nc{\et}{\beta} 
\nc{\gam}{\gamma}
\nc{\ka}{\kappa}
\nc{\lam}{\lambda}
\nc{\Lam}{\Lambda}

\nc{\ta}{\alpha}
\nc{\w}{\omega}
\nc{\io}{\iota}
\nc{\s}{\sigma}
\nc{\vphi}{\varphi}

\nc{\e}{\epsilon}

\renewcommand{\k}{v}

\nc{\ran}{\rangle}
\nc{\lan}{\langle}

\newcommand{\im}{{\rm Im}}

\nc{\bfone}{{\bf 1}}
\nc{\1}{{\bf 1}}

\nc{\dd}{\mathrm{d}}

\newcommand{\DETAILS}[1]{}

\newcommand{\Cp}{\mathrm{c}}

\DeclareMathOperator{\dG}{\mathrm{d}\Gamma}

\newcommand{\Norm}[1]{{\left\vert\kern-0.25ex\left\vert\kern-0.25ex\left\vert #1 
		\right\vert\kern-0.25ex\right\vert\kern-0.25ex\right\vert}}

\DeclareMathOperator{\conv}{conv}

\begin{document}
	
	
	\title{Macroscopic suppression of supersonic quantum transport}
	

	\author{J\'er\'emy Faupin} 
    	\altaffiliation{These authors contributed equally to this work.}
    \affiliation{Institut Elie Cartan de Lorraine, Université de Lorraine, 57045 Metz Cedex 1, France}

    \email{jeremy.faupin@univ-lorraine.fr}
	
	\author{Marius Lemm} \affiliation{Department of Mathematics, University of T\"ubingen,
		72076 T\"ubingen, Germany} \email{marius.lemm@uni-tuebingen.de}
	
	\author{Israel Michael Sigal} \affiliation{Department of Mathematics, University of Toronto, Toronto, M5S 2E4, Ontario, Canada} \email{im.sigal@utoronto.ca}
	
	\author{Jingxuan Zhang
    		(\begin{CJK*}{UTF8}{gbsn}张景宣
		\end{CJK*})
    } \affiliation{Yau Mathematical Sciences Center, Tsinghua University, Haidian District, Beijing 100084, China} \email{jingxuan@tsinghua.edu.cn}

	\date{October 15, 2025}
	
	\pagestyle{plain}
	
	\begin{abstract}
		We consider a broad class of strongly interacting quantum lattice gases, including the Fermi-Hubbard and Bose-Hubbard models. We focus on macroscopic particle clusters of size  $\theta N$, with $\theta\in(0,1)$ and $N$ the total particle number, and we study the quantum probability that such a cluster is transported across a distance $r$ within time $t$. Conventional effective light cone arguments yield a bound of the form $\exp(v t-r)$. We report a substantially stronger bound $\exp(\theta N(vt-r))$, which provides exponential suppression that scales with system size. Our result establishes a  universal dynamical large deviation principle:  macroscopic suppression of supersonic macroscopic   transport (MASSMAT). 

	\end{abstract}
	
	\maketitle

	Lieb and Robinson \cite{lieb1972finite} famously discovered that quantum lattice systems exhibit an ``effective light cone'' reminscient of relativistic systems.
	Their Lieb-Robinson bound (LRB) controls the probability that quantum information travels a distance $r>0$ in time $t>0$ by
	\begin{equation}\label{eq:LRB}
		\exp(C(v_{\mathrm{LR}}t-r))
	\end{equation}
	for constants $C,v_{\mathrm{LR}}>0$. This establishes an effective light cone $v_{\mathrm{LR}}t=r$ beyond which information propagation is exponentially suppressed. The Lieb-Robinson velocity $v_{\mathrm{LR}}$ is an $\mathcal O(1)$ quantity proportional to the strength of local interactions.

	As one of the few rigorous tools for analyzing strongly interacting quantum many-body systems, the LRB has proven remarkably powerful.  Following breakthroughs of Hastings in the early 2000s \cite{hastings2004lieb,hastings2005quasiadiabatic,hastings2007area}, it was decisive in resolving fundamental problems across condensed matter physics, quantum information theory, and high-energy physics. Applications of LRBs include exponential clustering  for gapped systems \cite{hastings2004lieb,nachtergaele2006lieb}, the definition and stability of topological quantum phases \cite{hastings2004lieb,hastings2005quasiadiabatic,bachmann2012automorphic,bravyi2010topological,nachtergaele2022quasi,yin2024low,de2024ldpc},  the area law for the entanglement entropy \cite{hastings2007area}, the control of dynamical entanglement generation \cite{bravyi2006lieb,tran2021optimal}, the many-body adiabatic theorem \cite{bachmann2017adiabatic,monaco2019adiabatic}, quantum simulation algorithms \cite{kliesch2014lieb,woods2015simulating,haah2021quantum,tong2022provably,kuwahara2024effective},  bounds on quantum messaging \cite{epstein2017quantum}, and the fast scrambling conjecture \cite{lashkari2013towards,roberts2016lieb,tran2021lieb,xu2024scrambling,lemm2025out}. Given the broad utility of LRBs, a large and continually growing body of research is concerned with extending them and related propagation bounds to new settings, e.g., to long-range interactions \cite{hastings2004lieb,nachtergaele2006lieb,chen2019finite,kuwahara2020strictly,tran2021lieb,tran2021optimal}, open quantum systems \cite{poulin2010lieb,nachtergaele2011lieb,mobus2023dissipation,breteaux2024light}, bosonic lattice gases \cite{nachtergaele2009lieb,schuch2011information,woods2015simulating,wang2020tightening,yin2022finite,faupin2022lieb,lemm2023information,lemm2023microscopic,kuwahara2024effective,kuwahara2024enhanced,lemm2024enhanced,lemm2024local}, and continuum systems \cite{gebert2020lieb,hinrichs2024lieb,bachmann2024lieb}.
	Improved LRB establishing finer control (e.g., slow transport for disordered systems) have also been proved \cite{
		hamza2012dynamical,damanik2014new,gebert2016polynomial,baldwin2023disordered,toniolo2024stability,elgart2024slow,yin2024low}.
	LRBs have also been observed experimentally \cite{cheneau2012light,them2014towards,richerme2014non,cheneau2022experimental}, e.g., with ultra-cold atoms in optical lattices. For a comprehensive review of progress up to 2023, see \cite{chen2023speed}.

	Ordinarily, it is considered a strength of the standard LRB \eqref{eq:LRB} that it is independent of system size, making it well-suited for analyzing quantum dynamics on microscopic scales, where all relevant parameters are $\mathcal O(1)$. 
	However, many physically relevant problems concern the collective transport of  \textit{macroscopic numbers of quantum particles}, starting with Ohm's law and ranging to the separation of timescales that is the basis of quantum hydrodynamics \cite{wyatt2005quantum,lucas2015hydrodynamic,lucas2018hydrodynamics} and prethermalization phenomena \cite{gring2012relaxation,mori2018thermalization,mallayya2019prethermalization}.   Controlling macroscopic particle transport poses unique challenges --- particularly in bosonic systems \cite{hamazaki2022speed,faupin2022maximal,lemm2023information,van2023topological,van2024optimal,li2025macroscopic} which can exhibit large local particles numbers even within regions of $\mathcal O(1)$ size.

	In this Letter, we establish a new type of dynamical bound on the transport of macroscopic particle clusters in strongly interacting quantum lattice systems. Specifically, we show that such transport is   suppressed by an exceptionally rapid and macroscopically large decay rate outside of a light cone: the bound takes the form
	\begin{equation}\label{eq:our}
		\exp(CN(v t-r)),
	\end{equation}
	where $N$ is the total particle number. Figure \ref{fig:lccompare} compares the standard light cone  to the new light cone given by \eqref{eq:our}. What sets the latter apart is the $N$-factor in the exponent in  \eqref{eq:our}. Consequently, the exponential decay rate outside of the light cone $r>vt$  grows  {extensively} with the system size $N$.	We refer to the bound \eqref{eq:our} as a manifestation of a new quantum-dynamical large deviation principle: \textit{macroscopic suppression of supersonic macroscopic   transport} (MASSMAT). It  significantly strengthens prior bounds on the macroscopic particle transport problem  \cite{hamazaki2022speed,faupin2022maximal,lemm2023information,van2023topological,van2024optimal,li2025macroscopic} for a broad class of quantum many-body Hamiltonians with short-ranged hopping.
	
	A key conceptual consequence of MASSMAT is that the effective light cone $r=vt$ established by \eqref{eq:our} becomes extremely sharp already for moderate $N$-values and mathematically exact (meaning free from errors) in the thermodynamic limit $N\to\infty$ --- all while $r$ and $t$ are held fixed.   This stands in contrast to the standard Lieb-Robinson light cone of the form \eqref{eq:LRB}, which is rougher because it allows for $\mathcal O(1)$ leakage. 
	Hence, MASSMAT establishes that the transport of  macroscopic particle cluster is universally governed by an unforeseen \textit{emergent strict causality}. The well-known analogy connecting LRBs and special relativity through the shared concept of the light cone is thus shown to become an exact correspondence for thermodynamically large particle clusters thanks to MASSMAT. The precise, rigorous statement is given in Theorem \ref{thm:MVB} below.
	
	Our proof of MASSMAT applies broadly to many strongly interacting quantum lattice gases,  including the Fermi-Hubbard and Bose-Hubbard Hamiltonians with short-range hopping \footnote{In fact, the lattice structure is not essential either: what we use abstractly is the short-range interaction and bounded group velocity of the kinetic operator. These requirements also hold, e.g., for semi-relativistic electrons in the continuum which are described by the kinetic operator $\sqrt{-\Delta+m^2}$.}.  Therefore, MASSMAT is a \textit{universal} dynamical principle that places unforeseen constraints on quantum many-body systems out of equilibrium.

	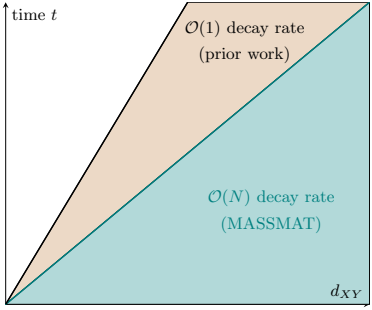
\begin{figure}
		\resizebox{5cm}{!}{
			\begin{tikzpicture}
				\begin{axis}[
					axis lines = middle,
					xlabel={$d_{XY}$}, ylabel={time $t$},
					xmin=0, xmax=5, ymin=0, ymax=5,
					xtick=\empty, 
					ytick=\empty,
					legend pos=north east
					]

					\addplot[fill=brown!30, opacity=0.5] 
					coordinates {(0,0) 
						(1,1) (2,2) (2.5,2.5) (3,3) (3.5,3.5) (4,4) (4.5,4.5) (5,5) 
						(5,5) (4.5,5) (4,5) (3.5,5) (3,5) (2.5,5) (2,4) (1,2) (0,0)} --cycle;
					
					\addplot[fill=green!30, opacity=0.5] 
					coordinates {(0,0) 
						(1,1) (2,2) (2.5,2.5) (3,3) (3.5,3.5) (4,4) (4.5,4.5) (5,5) 
						(5,0) (0,0)} --cycle;
					\addplot[thick, green] {x};
					
					\addplot[thick, , domain=0:2.5] {2*x};
					
				\end{axis}
				\node[] at (4.5,5.2) {$\mathcal{O}(1)$ decay rate};
				\node[] at (4.5,4.7) {(prior work)};
				\node[green] at (5,2) {$\mathcal{O}(N)$ decay rate};
				\node[green] at (5,1.5) {(MASSMAT)};
			\end{tikzpicture}
		}
		\caption{Our main result establishes the green light cone $\sim v t$, with $v$ given by \eqref{SR}, outside of which the exponential decay rate becomes $\propto N$, i.e., extensive. Since $v$ is larger than the quantity $\kappa$ in \eqref{eq:oldkappa} that bounded the speed of macroscopic clusters in prior work \cite{faupin2022maximal,lemm2023information}, there is a separation of the new macroscopic MASSMAT light cone and the usual $\mathcal O(1)$ light cone (yellow region).  Note that we establish the MASSMAT light  cone for short-ranged hopping terms, whereas \cite{hamazaki2022speed,faupin2022maximal,lemm2023information,van2023topological,van2024optimal,li2025macroscopic}  considered long-ranged hopping terms as well.}
		\label{fig:lccompare}
	\end{figure}
	
	
	

	\textit{Example: non-interacting chain.} For illustration, we present a simple situation where MASSMAT obviously holds, while the decay provided by the LRB \eqref{eq:LRB} is far too pessimistic. Consider the dynamics of a chain of free (i.e., non-interacting) bosons with only nearest-neighbor hopping, i.e., the Hamiltonian
	\begin{equation}\label{eq:Hfree}
		H_{\mathrm{free}}=\sum_{x=1}^{L-1} (a_x^\dagger a_{x+1}+a_{x+1}^\dagger a_{x}),
	\end{equation}
	where $\{a_x^\dagger,a_x\}_{x\in\Lambda}$ are the bosonic creation and annihilation operators.	For simplicity, consider the initial state where all particles are localized at site $1$, i.e., $\psi_0=(a^\dagger_1)^N\Omega$ where $\Omega$ is the vacuum. We capture  macroscopic transport via the projection $P_{N_{\{r,\ldots,L\}}\geq \theta N}$ onto the eigenspaces of $N_{\{r,\ldots,L\}}=\sum_{x=r}^L n_x$ (the number of particles sitting on sites $\{r,\ldots,L\}$)  with eigenvalues at least $\theta N$, $0<\theta<1$.  Using that the particles are non-interacting, one easily finds that for all $t,r\geq 1$ (see \cite{SM} for the details)
	\begin{equation}\label{eq:MASSMATfree}
		\langle \psi_t| P_{N_{\{r,\ldots,L\}}\geq \theta N}
		|\psi_t\rangle
		\leq  e^{\theta N C(v t-r)},
	\end{equation}
	which proves that the MASSMAT principle holds for the non-interacting chain. 
	
	For non-interacting particles, the power $N$ in \eqref{eq:MASSMATfree} arises because the particles are statistically independent. Of course, this argument breaks down completely for strongly interacting particles. Surprisingly, as we show, MASSMAT holds nonetheless. 
	We are able to achieve this by devising a new way of deriving many-body propagation bounds that we coin  \textit{geometric exponential tilting}, which is different from prior approaches to bounding macroscopic transport \cite{hamazaki2022speed,faupin2022maximal,lemm2023information,van2023topological,van2024optimal,li2025macroscopic}. We explain the core idea of geometric exponential tilting after we present the main result.

	\textit{Setup and main result.}
	We consider a finite graph $(\Lambda,\mathcal E_\Lambda)$ with vertex set $\Lambda\subset \mathbb{R}^D,\,D\ge1$, such that nearest-neighbors all have Euclidean distance $=1$.
	
	We consider a system of indistinguishable quantum particles living on the graph $\Lambda$ --- the result and its proof are identical for fermions and bosons and so we treat both cases in parallel. We take $\{a_x^\dagger,a_x\}_{x\in\Lambda}$ to be a collection of fermionic/bosonic creation and annihilation operators satisfying the usual canonical anticommutation/commutation relations.

	\textit{The model.}
	On the fermionic/bosonic Fock space over the one-body Hilbert space $\ell^2(\Lambda)$, we consider Hamiltonians of the form 
	\begin{align}
		\label{BH}
		{H = H_0  + V(\{n_x\}_{x\in\Lambda}),\quad H_0=\sum_{x,y \in \Lambda} J_{xy} a_x^\dagger a_y}.
	\end{align}
	Here, $J_{xy}$ represents short-range particle hopping and $V(\{n_x\}_{x\in\Lambda})$ represents a general density-density interaction. 
	
	Our assumptions on the Hamiltonian \eqref{BH} are as follows:

	\begin{enumerate}[label=(\roman*)]
		\item \label{C1}  The hopping matrix is Hermitian, i.e., $J_{xy}=\bar J_{yx}$ for $x,\,y\in\Lam$, and satisfies, for a parameter $a>0$, the \textit{short-range condition} that
		\begin{align}
			\label{SR}
			v=\max_{x\in\Lambda}\sum_{y\in\Lambda}|J_{xy}| \frac{\sinh(a|x-y|)}{a}
		\end{align} 
		is bounded independently of $|\Lambda|$.
		\item \label{C2}   {$V: \{0,1,2,\ldots\}^{|\Lambda|}\to \mathbb R$ is a real-valued function of  $|\Lambda|$ variables.}
	\end{enumerate}
	Under these assumptions, $H$ is a self-adjoint operator on the Fock space $\mathcal{F}(\ell^2(\Lambda))$; see, e.g., \cite{faupin2022lieb}.
	
	The class of Hamiltonians of the form \eqref{BH}  satisfying these assumptions is very broad. In particular, it includes the paradigmatic Fermi-Hubbard and Bose-Hubbard Hamiltonians. We call Condition \ref{C1} the short-range condition because 
	\[
	\sinh(a|x-y|)\sim \frac{1}{2}\exp(a|x-y|), \quad |x-y|\gg 1
	\]
	and so $v$ in \eqref{SR} is bounded independently of $\Lambda$ precisely when the hopping matrix $J_{xy}$ decays exponentially at large distances $|x-y|$. (We use the $\sinh$ in \eqref{SR} instead of the exponential because it gives the asymptotically sharp value of $v$ in the limit $a\to 0$, as we explain after the main result.)
	In particular, Condition \ref{C1} holds for the physically most important case of nearest-neighbor hopping on the integer lattice $\Lam\subset\Zb^D$,   i.e.,
	\[
	J_{xy}=J\delta_{|x-y|=1},
	\]
	in which case the the short-range condition \eqref{SR} holds for any  $a>0$ with $v =2DJ \tfrac{\sinh a}{a}$. 
	
	Assumption  \ref{C2} on $V$ is extremely weak. In particular, long-range and $k$-body interactions for any $k$ are allowed as long as they are of density-density type. Typical examples of $V(\{n_x\}_{x\in\Lambda})$ are polynomials. E.g., for the paradigmatic Bose-Hubbard Hamiltonian, one has
	\[
	V(\{n_x\}_{x\in\Lambda})=\sum_{x\in\Lambda} \left( n_x(n_x-1)-\mu n_x\right).
	\]
	We remark that it is easy to include local spin degrees of freedom in our setup and we only refrain from doing so to keep the notation simple. Spin degrees of freedom appear, e.g., in the standard Fermi-Hubbard Hamiltonian.
	We can also treat time-dependent Hamiltonians, i.e., $J_{xy}=J_{xy}(t)$ and $V=V(t)$. In this case, we simply require that Assumptions \ref{C1} and \ref{C2} hold uniformly in $t$. 
	
	\textit{The main result.}
	Since the Hamiltonian \eqref{BH} preserves the total particle number $N_\Lambda =\sum_{x\in\Lambda}n_x$, we henceforth work on a fixed eigenspace of $N_\Lambda=N$ for a fixed $N\geq 1$.	
	To state our main result, we introduce some notation.
	For any subset $S\subset \Lam$, we define  
	\begin{equation} \label{eq:Nlocdefn}
		N_S=\sum_{x\in S}n_x,\qquad \bar N_S=\frac{N_S}{N}.
	\end{equation}
	For $0\le c\le 1 $, we write  $P_{\bar N_{S}\geq c} $ for the associated spectral projector of $\bar N_S$. 
	Finally, given two subsets of the lattice $X,Y\subset \Lam$, we write $d_{XY}$ for their Euclidean distance. 
	
	Our main result is the following:

	\begin{theorem}[MASSMAT principle]\label{thm:MVB}
		Consider a Hamiltonian $H$ of the form \eqref{BH} satisfying Assumptions \ref{C1}--\ref{C2} with $v,\,a>0$. 
		
		Then, for any $0\le \alpha< \beta\le1$,  and any disjoint subsets $X,\,Y\subset \Lam$,   the following estimate holds on each $N$-particle sector:
		\begin{equation} \label{MVB} 
			\|P_{ \bar N_X\ge \beta}\,e^{-\mathrm{i} tH}P_{ \bar N_{Y}\ge 1-\alpha}  \|\leq \,  e^{-a N((\beta-\alpha)d_{XY}- v|t|)}.
		\end{equation} 
	\end{theorem}

	This theorem is proved in \cite{SM}.
	The bound \eqref{MVB} implies a  strong light cone estimate on the quantum probability that a macroscopic cluster comprised of $(\beta-\alpha)N$ particles traverses the distance $d_{XY}$ in time $t$. Indeed, consider two disjoint regions $X,Y\subset \Lambda$ and an initial $N$-particle density operator $\rho_0$ that has at least $(1-\alpha) N$ particles in $Y$ and thus satisfies $\Tr(P_{ \bar N_{Y}\ge 1-\alpha}\rho_0)=1$. Consequently, there are at most $\alpha N< \beta N$ particles in $X$ initially and so  $\Tr(P_{ \bar N_X\ge \beta}\rho_0)=0$.  
	We denote the time-evolved state by $\rho_t=e^{-\mathrm{i}tH}\rho_0e^{\mathrm{i}tH}$. Then $\Tr(P_{ \bar N_X\ge \beta}\rho_t)$ is the probability that after time $t$, at least $(\beta-\alpha)N$ particles are transported from region $Y$ to the region of interest $X$. Thanks to \eqref{MVB}, this probability is bounded by
	\begin{align}
		&\Tr(P_{ \bar N_X\ge \beta}\rho_t)	\notag\\=&\Tr(P_{ \bar N_X\ge \beta} e^{-\mathrm{i} tH}P_{ \bar N_{Y}\ge 1-\al}\rho_0P_{ \bar N_{Y}\ge 1-\al}e^{\mathrm{i} tH}P_{ \bar N_X\ge \beta})\notag\\
		\le&\|P_{ \bar N_X\ge \beta} e^{-\mathrm{i} tH}P_{ \bar N_{Y}\ge 1-\al}\|^2\Tr(\rho_0)\notag\\
		\leq& e^{-2a N((\beta-\alpha)d_{XY}- v|t|)}.\label{e9}\tag{*}
	\end{align}
	
	In words, \eqref{MVB} implies that a macroscopic cluster of $(\beta-\al) N$ particles move at most at speed $\frac{v}{\beta-\al}$, up to errors that are exponentially small in $N$ and thus effectively completely negligible.

	A few remarks about the bound \eqref{MVB} are in order. First, since it is an operator norm bound, it provides state-independent constants. 
	Second, as mentioned above, the propagation speed (i.e., the slope of the MASSMAT light cone) is given by  $\frac{v}{\beta-\alpha}$. The constant $v$ is related to previous velocity bounds on particle transport \cite{faupin2022maximal,lemm2023information} as follows. Since $\frac{\sinh z}{z}\geq 1$ for $z\geq0$, we have that Assumption \eqref{SR} implies
	\begin{equation}\label{eq:oldkappa}
		\kappa=\max_{x\in\Lambda}\sum_{y\in\Lambda}|J_{xy}||x-y| < \k.
	\end{equation}
	This $\kappa$ is exactly the first moment of the hopping matrix which was used to bound the propagation speed in our prior works \cite{faupin2022maximal,lemm2023information}.
	Thus \eqref{eq:oldkappa} shows that the maximal velocity / light cone slope $v$ for our light cone here is slightly larger than the slope $\kappa$ obtained in \cite{faupin2022maximal,lemm2023information}. This shows that the macroscopic decay rate outside of the MASSMAT light cone (which has slope $v$) sets on slightly later than the standard $\mathcal O(1)$ decay (compare \eqref{eq:LRB}) that was proved in \cite{faupin2022maximal,lemm2023information}. This is shown in \figref{fig:lccompare}.
	In fact, our choice of $v$ in Condition \ref{C1} is sharp in the limit of arbitrarily slow decay $a\to 0$: Using $\sinh(a|x-y|)\sim a|x-y|$, we see that $v$ in \eqref{SR} converges to $\kappa$ as $a\to 0$.
	Compared to \cite{faupin2022maximal,lemm2023information}, we see that by slightly increasing the light cone slope, we are able to boost the microscopic error estimate outside of the light cone to an unprecedented, macroscopic one.

    For finite-range hopping, Assumption \ref{C2} holds for any $a>0$, and so it is possible to optimize the choice of $a$ depending on the other parameters. Consider, e.g., nearest neighbor hopping on an integer lattice, i.e,  $J_{xy}=J\delta_{|x-y|=1}$. Then, the minimizer  is $a_*= \cosh^{-1}\left(\frac{(\beta-\al) d_{XY}}{2DJ\abs{t}}\right)$ and it yields an improved bound of the form $\left(\frac{\abs{t}}{(\beta-\alpha)d_{XY}}\right)^{N(\beta-\alpha)d_{XY}}$. This is   a MASSMAT strengthening of the refined LRB of the form $( t/d)^d$ which  recently played a crucial role in achieving refined control over dynamical entanglement generation \cite{toniolo2024dynamical}. 
    
		  Incorporating physical units in  our theorem amounts to replacing  $J\to \frac{J}{\hbar}$ and $d_{XY}\to \ell r_0$ with $r_0$ the lattice spacing and $\ell$ an $\mathcal O(1)$ number.
        Let us consider a typical 1D optical lattice experiment realizing  the Bose-Hubbard Hamiltonian, e.g., \cite{cheneau2012light,cheneau2022experimental}, which features $N=18$ atoms with an effective hopping amplitude $J/\hbar \approx 500  \mathrm{s}^{-1}$ between neighboring lattice sites that are spaced $r_0\approx 500 \mathrm{nm}$ apart, observed up to time $t_{\max} \approx 3 \hbar/J$. We aim to bound, say, the quantum probability that $1/3$ of the $N=18$ particles are transported across $\ell$ lattice sites in time $t$. We apply our theorem with the dimensionally correct choice $a=1/r_0$ and use $\sinh(1)\leq 6/5$ to obtain  the bound
		\[
		\exp\left(-N\left(\frac{\ell}{3}-\frac{3J}{\hbar }t \right)\right).
		\]
		Experimentally, the interior quantity $\frac{\ell}{3}-\frac{3J}{\hbar }t$ is of order one; e.g., taking $\ell= 6$ and $t=\frac{1}{9}t_{\max}=\frac{1}{3}\hbar/J$, we have $\frac{\ell}{3}-\frac{3J}{\hbar }t=1$. Then the extra factor of $N=18$ improves the probability bound from $e^{-1}\approx 0.37$ to $e^{-18}\approx 1.52 \times 10^{-8}$.

	\textit{Description of the proof method.} Our proof of Theorem \ref{thm:MVB} rests on {an} approach which we call \textit{geometric exponential tilting}. Here, we give a high-level overview and compare the approach to other ones in the literature. The full proof is deferred to \cite{SM}. Geometric exponential tilting is completely different to the approaches used in prior works to bound transport of  macroscopic boson clusters, namely the second-order adiabatic spacetime localization observables (ASTLO) method \cite{faupin2022maximal,lemm2023information} and the optimal transport method \cite{van2024optimal}. 
	The overarching idea of geometric exponential tilting is simple: We introduce a suitably chosen, invertible (but not unitary) many-body similarity transformation $T$, and then we bound the left-hand side of \eqref{MVB} by 
	\begin{equation}\label{eq:proofidea}
		\begin{aligned}
			&\|P_{ \bar N_X\ge \beta}\,e^{-\mathrm{i} tH}P_{ \bar N_{Y}\ge 1-\alpha}  \|\\
			\leq&  \|P_{ \bar N_X\ge \beta} T^{-1}\|\|T e^{-\mathrm{i} tH}T^{-1}\| TP_{ \bar N_{Y}\ge 1-\alpha}  \|.
		\end{aligned}
	\end{equation}
	The first and third norms will produce the spatial decay $e^{-a N(\beta-\alpha)d_{XY}}$. The middle norm of $T e^{-\mathrm{i} tH}T^{-1}$ (which we call the deformed propagator) will produce the growth in time $e^{a Nv|t|}$ for $t\in\Rb$, and so \eqref{MVB} follows.

	The crux, of course, lies in choosing the right similarity transformation $T$. We construct a $T$ that exponentially weights the local particle numbers in a site-dependent way, i.e.,  $T=\exp(\sum_{x} F(x) n_x)$ for a suitable real-valued function $F(x)$. The function $F(x)$ interpolates continuously
	between being $1$ on the region $X$ and $=-1$ on the region $Y$. Thus, $T$ gives large weight to configurations with many particles in $X$ and small weight to configurations with many particles in $Y$. The bound $\|T e^{-\mathrm{i} tH}T^{-1}\| \leq e^{a Nv|t|} $  shows that the exponential weights in $T$ grow at most exponentially in time under the dynamics.    
	
	The geometric exponential tilting method is simultaneously conceptually simple (recall \eqref{eq:proofidea}), flexible (one can adapt the similarity transform $T$ to the problem at hand) and powerful (it is so far the only method that yields MASSMAT). 

	The method has various links to prior works. First, it is inspired by a recent complex analysis argument for deriving transport bounds on non-interacting particles in \cite{SW}; see also \cite{cedzich2024exponential}. The connection to complex analysis arises through Paley-Wiener theory \cite{paley1934fourier}, which in particular says that it is equivalent to have exponential decay in position space and to have an analytic extension of the Fourier transform to a complex strip \footnote{Indeed, 
		Condition \ref{C1} is equivalent to saying that the family of transformed Hamiltonians $H_\xi=T_\xi H T_\xi^{-1}$, where $T_\xi=e^{\dG(\xi\cdot x)}$, has an analytic continuation, $H_\z$, into the strip $\cS_a^n=\Set{\z\in\Cb^n:\abs{\im \z_j}<a\;\forall j}$.
		This analytic continuation breaks the time-reversal symmetry and one takes $\im \z_j>0$ or $\im\z_j<0$ $\forall j$, depending on whether one considers $t>0$ or $t<0$.}.
	From a broader perspective, using suitable similarity transforms with locally varying exponential weights to adapt the geometry to the question at hand has a long history in mathematical physics, perhaps most famously in Witten's proof of the Morse inequalities \cite{witten1982supersymmetry}. 
	In the context of propagation bounds on quantum many-body systems, related uses of spatially varying weights have recently appeared in Yin-Lucas \cite{yin2022finite}, Osborne-Yin-Lucas \cite{osborne2024locality}, and Fresta-Porta-Schlein \cite{fresta2024effective} for different quantum-dynamical problems.

	\textit{Conclusions.}	
	In this work, we have identified and rigorously proven a conceptually novel, universal bound on the  nonequilibrium dynamics of strongly interacting quantum lattice models: the macroscopic suppression of supersonic macroscopic transport (MASSMAT). MASSMAT is an unforeseen dynamical large deviation principle, which establishes that the quantum probability of supersonic propagation of macroscopic particle numbers actually decays exponentially at a macroscopic rate proportional to the total particle number $N$. This is in stark contrast to what one obtains from Lieb-Robinson bounds, which give an $\mathcal O(1)$ decay rate that does not grow with $N$. MASSMAT substantially strengthens the decay rate achieved on macroscopic boson transport in prior works \cite{faupin2022maximal,lemm2023information,van2023topological,van2024optimal,li2025macroscopic}.
	
	The MASSMAT principle is universal in scope: It applies to both bosons and fermions (as well as mixtures) and holds across general geometries.   	
	Our proof is based on a new analytical technique --- geometric exponential tilting --- that is inspired by complex analysis methods from one-body quantum mechanics and developed here for the first time in a many-body context. We anticipate that this method will find broader applications in macroscopic transport problems, especially in regimes characterized by slow transport of large clusters, such as hydrodynamic limits \cite{wyatt2005quantum,lucas2015hydrodynamic,lucas2018hydrodynamics} or prethermalization phenomena \cite{gring2012relaxation,mori2018thermalization,mallayya2019prethermalization}.

	Our work opens several avenues for future exploration. One key question is the experimental observation of MASSMAT, e.g., in ultracold  quantum gases on optical lattices. This requires  observation of particle numbers that are large enough so that the improved decay outside of the MASSMAT light cone becomes observable. Another important avenue is to investigate if the MASSMAT principle extends to systems with long-range hopping, as studied in \cite{faupin2022maximal,lemm2023information,van2023topological,van2024optimal,li2025macroscopic}.

	\section*{Data availability}
	Data sharing is not applicable to this article as no datasets were generated or analyzed during the current study.
	
	\section*{Acknowledgments} The authors thank Ryusuke Hamazaki, Tomotaka Kuwahara, and Tan Van Vu for useful comments on a preprint version of the manuscript. The research of J.F. is supported by the ANR, project ANR-22-CE92-0013.
	The research of M.L.\ is supported by the DFG through the grant TRR 352 – Project-ID 470903074 and by the European Union (ERC Starting Grant MathQuantProp, Grant Agreement 101163620). I.M.S.~is supported by NSERC Grant NA7901. 
	J.Z.~is supported by National Key R \& D Program of China Grant 2022YFA100740, China Postdoctoral Science Foundation Grant 2024T170453, National Natural Science Foundation of China Grant 12401602, and the Shuimu Scholar program of Tsinghua University. 
	
	\bibliography{MVBMBwExpDecayBib.bib}

	\widetext
	\pagebreak
	\begin{center}
		\textbf{\large Supplemental Material:\\ Macroscopic suppression of supersonic quantum transport}
		
		\vspace{.5cm}
		
		J\'er\'emy Faupin, Marius Lemm, Israel Michael Sigal, and Jingxuan Zhang
	\end{center}

	\stepcounter{myequation}
	\setcounter{figure}{0}
	\setcounter{table}{0}
	\makeatletter
	\renewcommand{\theequation}{S\arabic{equation}}
	\renewcommand{\thefigure}{S\arabic{figure}}

	This appendix has two parts.  In Part I, we give the short proof of MASSMAT for the special case of a chain of non-interacting bosons \eqref{eq:Hfree}, which was displayed in the main text as \eqref{eq:MASSMATfree}. In Part II, we introduce the geometric exponential tilting method and give the full proof of our main result, Theorem \ref{thm:MVB}.

	\section{I. Direct proof of MASSMAT for non-interacting bosons}\label{app:I}
	In this appendix, we prove that \eqref{eq:MASSMATfree} holds for the non-interacting Hamiltonian \eqref{eq:Hfree} by a short calculation.
	We consider the more general  initial state $\psi_0=(a^\dagger(f))^N\Omega$ where $\Omega$ is the vacuum and $f:\{1,\ldots,L\}\to \mathbb C$ is a one-body wave function which is localized around the origin. 
	Since the particles are non-interacting,
	\[
	\psi_t=e^{-\mathrm{i}tH}\psi_0=(a^\dagger(e^{-\mathrm{i}t\Delta_L}f))^N\Omega
	\]
	and so, 
		\begin{equation}\label{eq:freetime}
			\begin{aligned}				\langle \psi_t| P_{N_{\{r,\ldots,L\}}\geq \theta N}
				|\psi_t\rangle
				=&\sum_{N'=\lceil \theta N\rceil }^N \binom{N}{N'}\left(\sum_{x=r}^L|\langle e^{-\mathrm{i}t\Delta_L}f,\delta_x  \rangle |^{2}\right)^{N'}\left( 1-\sum_{x=r}^L|\langle e^{-\mathrm{i}t\Delta_L}f,\delta_x  \rangle |^{2}\right)^{N-N'}\\
				\le & 2^N \left(\sum_{x=r}^L|\langle e^{-\mathrm{i}t\Delta_L}f,\delta_x  \rangle |^{2}\right)^{\lceil \theta N\rceil},
			\end{aligned}
	\end{equation}
where the last line follows since $\sum_{N'=0}^N\binom{N}{N'}=2^N$.
	For the one-body Laplacian $\Delta_L$, it is easy to check from Fourier theory that $|\langle e^{-\mathrm{i}t\Delta_L}f,\delta_x  \rangle |^2\leq e^{\tilde C(\tilde v t-x)}$ for suitable constants $\tilde C,\tilde v>0$. 
	Therefore, 
		\begin{equation}\label{eq:APP_MASSMATfree}
			\langle \psi_t| P_{N_{\{r,\ldots,L\}}\geq \theta N}
			|\psi_t\rangle
			\leq  e^{\lceil \theta N\rceil C(v' t-r)}.
		\end{equation}
		Here we used that, since we assume $t\geq 1$, various time-independent prefactors including $2^N$ can be absorbed in the velocity $v'$.
	
	The derivation can be adapted to include on-site external fields, i.e., to treat Hamiltonians of the form
	\[
	H_{\mathrm{free}}=\sum_{x=1}^{L-1} (a_x^\dagger a_{x+1}+a_{x+1}^\dagger a_{x}+v_x n_x),
	\]
	with $v_x$ given by a bounded sequence. For this, one uses the one-body propagation bound of the form $|\langle e^{-\mathrm{i}t(\Delta_L+V)}f,\delta_x  \rangle |^2\leq e^{C'(vt-x)}$, which follows, e.g., from \cite{cedzich2024exponential,SW}.


	\section{II. Proof of Theorem \ref{thm:MVB}}\label{app:II}
	
	The proof of \thmref{thm:MVB} is organized as follows.
	
	\begin{itemize}
		\item In Step 1, we render the relative geometry of two disjoint subsets $X$ and $Y$ effectively one-dimensional by constructing a ``separation function'' $s(x)$ that incorporates the relevant geometry. 
		\item In Step 2, we introduce the exponential tilting operator $T$, which involves a similarity transformation that exponentially weighs the local particle numbers in a site-dependent way. The relative geometry between $X$ and $Y$ is fully taken into account through the function $s(x)$ defined in \eqref{sDef1}, resp.\ \eqref{sDef}.
		\item In Step 3, we derive the spatial decay from the first and third terms in \eqref{eq:proofidea}.
		\item In Step 4, we bound the tilted deformed propagator, i.e., the middle term in \eqref{eq:proofidea} and conclude the proof.
	\end{itemize}

	\subsection*{Step 1: Separating functions}\label{sec21}

	\textit{Simple geometry: disjoint convex hulls.} To fix ideas, we first consider the simplified scenario of two subsets $ X,\, Y\subset \Lam$ whose convex hulls, $\tilde X=\conv(  X)$ and $\tilde Y=\conv(  Y)$, are disjoint. With a slight modification of the argument, the results extend to 	{complete general} disjoint subsets; see the construction \eqref{SOdef} and \eqref{sDef}.
	
	Let $x_0\in {\tilde X}$, $y_0\in {\tilde Y}$ be such that $d_{{\tilde X}{\tilde Y}}=|x_0-y_0|$, and introduce the ``center of mass'' coordinates
	\begin{equation*}
		w_0=\frac12(x_0+y_0), \qquad b=\frac{x_0-y_0}{|x_0-y_0|}.
	\end{equation*}
	By construction, the hyperplane $\Set{z\in\mathbb R^D:(z-w_0)\cdot b=0}$ separates ${\tilde X}$ and ${\tilde Y}$ to two different sides. Below we project the relative geometry of ${\tilde X}$ and ${\tilde Y}$ onto the line joining the points $x_0$ and $y_0$.
	
	We introduce the separating function $s:\mathbb R^D\to \mathbb R $, %
	\begin{align}
		\label{sDef1}
		s(x)=b\cdot (x-w_0) 
	\end{align} and define the following subsets of $\Lambda$ (see Figure \ref{figSep})
	\begin{equation}	\label{decomp}
		\begin{aligned}
			&Y^\infty=\big\{x\in\Lambda \, | \, s(x)\le-\frac12d_{{\tilde X}{\tilde Y}}\big\},\\
			&W^0=\big\{x\in\Lambda \, | \, -\frac12d_{{\tilde X}{\tilde Y}}<s(x)<\frac12d_{{\tilde X}{\tilde Y}}\big\},\\
			&X^\infty=\big\{x\in\Lambda \, | \, \frac12d_{{\tilde X}{\tilde Y}}\le s(x)\big\}.
		\end{aligned}
	\end{equation}

	\begin{figure} 
		\centering
		\begin{tikzpicture}[scale=.7]
			\coordinate (X0) at (3,0);
			\coordinate (Y0) at (-3,0);
			\coordinate (O) at (0,0);				
			\filldraw[gray!30] plot[smooth, tension=1] coordinates {(-5,1.6) (-3.8,1.8) (Y0) (-3.6,-1.8) (-5,-2)};
			
			\filldraw[gray!80] plot[smooth, tension=0.8] coordinates {(5,1.6) (3.7,1.5) (X0) (4,-2.2) (5,-2)};
			
			\draw[dashed,thick] (X0) -- ++(0,2.5);
			\draw[dashed,thick] (X0) -- ++(0,-2.5);
			
			\draw[dashed,thick] (Y0) -- ++(0,2.5);
			\draw[dashed,thick] (Y0) -- ++(0,-2.5);
			
			\draw[dashed] (O) -- ++(0,-2.5);				
			\draw[dashed] (O) -- ++(0,1.8);				
			\draw[->, thick] (O) -- node[above] {$\frac12{d_{\tilde X\tilde Y}}$} (X0);
			\draw[->, thick, blue] (O) -- node[below] {$b$} (0.4,0);
			\node[above left] at (O) {$w_0$};
			
			\node[left] at (-3.5,2.2) {$Y^\infty$};
			\node[left] at (0.5,2.2) {$W^0$};
			\node[left] at (4.5,2.2) {$X^\infty$};				
			
			\node[above left] at (Y0) {$y_0$};
			\node[above right] at (X0) {$x_0$};
			\node at (-4, 0) {${\tilde Y}$};
			\node at (4, 0) {${\tilde X}$};
			
		\end{tikzpicture}
		\caption{Schematic diagram for the decomposition \eqref{decomp}. }\label{figSep}
	\end{figure}
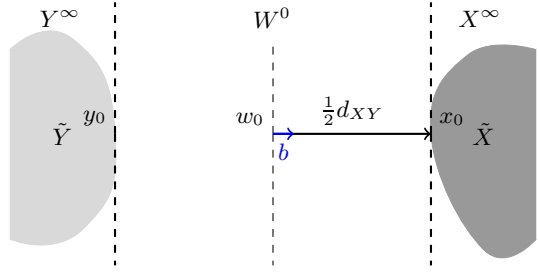

	For the simple geometry, we have the following easy lemma.
	\begin{lemma}\label{lem21}
		Assuming  $X$, $Y$ have disjoint convex hulls, we have ${\tilde X\subset X^\infty}$ and $\tilde Y\subset Y^\infty$.
	\end{lemma}
	\begin{proof}
		Let $x\in \tilde X$. We have
		\begin{align*}
			s(x)=&b\cdot (x-x_0)+b\cdot(x_0-w_0) \\=&\frac{x_0-y_0}{|x_0-y_0|}\cdot(x-x_0)+\frac12|x_0-y_0|.
		\end{align*}
		The second term equals $\frac12d_{{\tilde X}{\tilde Y}}$ by the choice of $x_0,\,y_0$, while the first term is non-negative by the separating plane theorem for disjoint convex sets \cite{luenberger1997optimization}. 
		This shows that $s(x)\ge\frac12d_{\tilde X\tilde Y}$ and hence that $\tilde X\subset X^\infty$. The proof that $\tilde Y\subset Y^\infty$ is analogous.
	\end{proof}

	\textit{Extension to general geometry.} \label{sec22}
	Consider now arbitrary   disjoint   subsets $X,Y\subset \Lambda$.
	Let $g(x)=\dist_Y(x)-\dist_X(x)$. The separating hyperplane is now replaced by the separating hypersurface
	\begin{align}
		\label{SOdef}
		S=\Set{g(x)=0},\quad  \Om_\pm = \Set{\pm g(x)>0} .
	\end{align} Indeed, the hypersurface $S$ is equidistant to $X$ and $Y$, with $X\subset \Om_+$, $Y\subset \Om_-$; see Figure \ref{figSep'}. 
	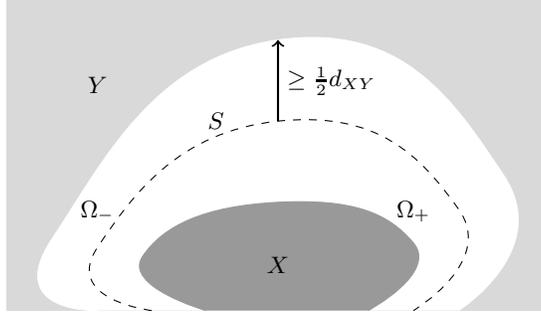
\begin{figure}[H] 
		\centering
		\begin{tikzpicture}[scale=.6]
			\fill[gray!30] (-6,0) rectangle (6,7);		
			\filldraw[white] plot[smooth, tension=1] coordinates {(-4,0) (-5,1.6) (0,6)   (5,3) (4,0) };
			\filldraw[gray!80] plot[smooth, tension=1] coordinates {(-1.6,0) (-3,1.2) (0,2.4)  (3,1.5) (2,0)};
			

			\draw[dashed] plot[smooth, tension=1] coordinates {(-2.8,0) (-4,1.4) (0,4.2) (4,2.25) (3,0)};
			
			\draw[->, thick] (0,4.2) -- node[right] {$\ge \frac12d_{XY}$} (0,6);

			
			
			
			
			\node[left] at (-1,4.2) {$S$};
			
			\node at (-4, 5) {$Y$};
			\node at (0, 1) {$X$};
			
			\node at (3, 2.2) {$\Om_+$};
			\node at (-4, 2.2) {$\Om_-$};
			
		\end{tikzpicture}
		\caption{Schematic diagram for $S$ and $\Om_\pm$. }\label{figSep'}
	\end{figure}

	In this more general case, we take the separating function $s(x)$ to be the signed distance function to $S$ with the sign chosen such that $\pm s(x)> 0$ on $\Om_\pm$.  Explicitly, 
	\begin{align}
		\label{sDef}
		s(x)=\sgn(g(x))\dist_S(x).
	\end{align} 
	(This reduces to \eqref{sDef1} in the simplified scenario considered before.) We note for later reference that this function is $1$-Lipschitz continuous, i.e., 
	\begin{align}
		\label{sLip}
		\abs{s(x)-s(y)}\le \abs{x-y},\quad x,\,y\in\Lam.
	\end{align}
	Indeed, on the same side of $S$, the function $s(x)$ coincides with the distance function up to a sign, which is 1-Lipschitz by the following standard argument. For any $z\in S$ and $x,\,y\in\Lam$, we have  $\dist_S(x)\leq |x-z|\leq |x-y|+|y-z|$ for any $S$. By taking $\inf _{z\in S}$, we obtain $\dist_S(x)\le |x-y|+\dist_S(y)$. If $x,y$ fall on different sides of $S$, then we join them with a line segment passing $S$ at, say, $z$, and then  apply the triangle inequality to $\dist_S(x) \le |x-z|$ to conclude.

	Similarly to \eqref{decomp}, we decompose $\Lam$ with $s(x)$ from \eqref{sDef} as follows:
	\begin{equation}\label{decomp'}	
		\begin{aligned}
			&Y^\infty=\big\{x\in\Lambda \, | \, s(x)\le-\frac12d_{{ X}{ Y}}\big\},\\
			&W^0=\big\{x\in\Lambda \, | \, -\frac12d_{{ X}{ Y}}<s(x)<\frac12d_{{ X}{ Y}}\big\},\\
			&X^\infty=\big\{x\in\Lambda \, | \, \frac12d_{{ X}{ Y}}\le s(x)\big\},
		\end{aligned} 
	\end{equation}
	As in \lemref{lem21}, we have
	\begin{lemma}\label{lemma:XYgeometry}
		We have $X\subset X^\infty$ and $Y\subset Y^\infty$.
	\end{lemma}
	\begin{proof}
		Consider any $x\in X$. On the one hand,  we have $s(x)=\dist_S(x)$  by Definition \eqref{sDef} and the fact that  $X\subset \Om_+$. On the other hand, since $S$ is equidistant to $X$ and $Y$, we have $d_{SX}= \frac12 d_{XY}$. 	
		%
		%
		This shows that $s(x)\ge d_{SX}=\frac12d_{XY}$ and hence that $X\subset X^\infty$. The proof of $Y\subset Y^\infty$ is analogous.
	\end{proof}

	\subsection*{Step 2: Exponential tilting operator}
	In this section, we introduce the exponential tilting operator $T$; see \eqref{Tdef'} below. 
	
	For brevity, we fix disjoint $X,\,Y\subset \Lam$ throughout and denote 
	\begin{equation}\label{dmDef}
		d=d_{XY}.
	\end{equation}
	We define the function $f:\mathbb R\to\mathbb R$ by %
	\begin{align}
		\label{fDef}
		f(s)=\1_{s\ge 1/2}-\1_{s\le -1/2} +2s \1_{\abs{s}<1/2}
	\end{align}
	as shown in \figref{fig:f}. (Here, $\1_{\ldots}$ is the indicator function which equals $1$ if condition $\ldots$ is satisfied and which equals zero otherwise.)
	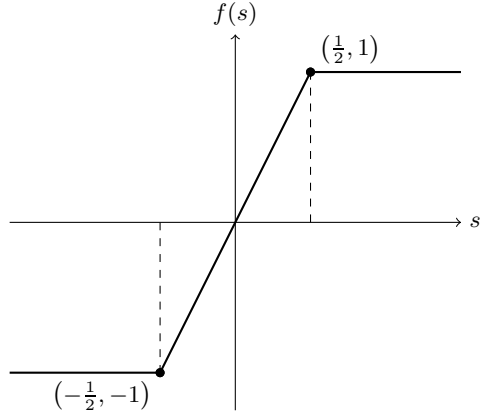
\begin{figure}
		\centering
		\begin{tikzpicture}[scale=2]
			\draw[->] (-1.5, 0) -- (1.5, 0) node[right] {$s$};
			\draw[->] (0, -1.25) -- (0, 1.25) node[above] {$f(s)$};
			
			\draw[thick] (-1.5, -1) -- (-0.5, -1);
			\fill(-0.5, -1) circle (0.03);
			
			\draw[thick] (-0.5, -1) -- (0.5, 1);
			\fill (-0.5, -1) circle (0.03);
			\fill (0.5, 1) circle (0.03);
			
			\draw[thick] (0.5, 1) -- (1.5, 1);
			\fill (0.5, 1) circle (0.03);
			
			\node[below left] at (-0.5, -1) {$\left(-\frac{1}{2}, -1\right)$};
			\node[above right] at (0.5, 1) {$\left(\frac{1}{2}, 1\right)$};
			
			\draw[dashed] (-0.5, 0) -- (-0.5, -1);
			\draw[dashed] (0.5, 0) -- (0.5, 1);
		\end{tikzpicture}
		\caption{The function $f(s)$. } \label{fig:f}
	\end{figure}	
	
	We use the function $f$ to truncate the signed distance function $s(x)$ in \eqref{sDef} to $f(\frac{s(x)}{d})$, with
	$d>0$ as in \eqref{dmDef}.  Notice that $f(\frac{s(x)}{d})= 1$ on $X$ and $f(\frac{s(x)}{d})=-1$ on $Y$; hence $f$ only acts as a distance in the white region in \figref{figSep'}.
	For this truncated distance function, we introduce, for all $x\in\Lambda$, 
	\begin{align}\label{qDef}
		q_\mu(x)=&\mathrm{exp}\Big(\mu f(\tfrac{s(x)}{d})\Big) >0,\\
		\label{eq:mudef}
		\mu=&\frac{da}{2},
	\end{align}  
	where the constant $a$ comes from the short range condition \eqref{SR}.  
	As usual, $q_\mu$ is identified with the corresponding multiplication operator.

	To lift $q_\mu$ to the Fock space $\mathcal{F}(\ell^2(\Lambda))$, it is convenient to introduce the following standard notation for the second quantization functor; see, e.g., \cite[p.~8]{BR} and \cite{reed1978iv}. 
	
	\begin{definition}[Second quantization functor]
		Given a one-body operator $A=(A_{xy})_{x,y\in\Lambda}$, we set
		\[
		\dG(A)=\sum_{x,y\in\Lambda} A_{xy} a_x^\dagger a_y,
		\]
		 {and we set  
			\[
			\Ga(e^A)=\exp(\dG(A)).
			\]}
	\end{definition}

	Note the following special case of this definition: If $A_{xy}=q(x)\delta_{xy}$ is a multiplication operator by a function $q(x)$, then
	\[
	\dG(q)=\sum_{x\in\Lambda} q(x) n_x.
	\]
	
	The reason for introducing the second quantization functor is that it allows to state various algebraic properties succinctly. Indeed, we will use the following properties which follow directly from the CAR/CCR \cite{reed1978iv,BR}.
	
	\begin{proposition}[Properties of $\dG$ and $\Gamma$]\label{proposition:Gammaproperties}\mbox{}
		
		\begin{itemize}
			\item[(i)] If $A\leq B$, then $\dG(A)\leq \dG(B)$. In particular, on each $N$-particle sector, we have $\dG(A)\leq \|A\| N$.
			\item[(ii)] For any function $q:\Lambda\to \mathbb C$ and any $x\in\Lam$, we have the pull-through formulas
			\begin{equation} \label{eq:commut0}
				\Gamma(q)a^\dagger _x=q(x)a^\dagger _x\Gamma(q), \quad a_x \Gamma(q) = \Gamma(q) \bar{q}(x) a_x.
			\end{equation}
		\end{itemize}
		
	\end{proposition}
	
	We can now define the central object of the proof.
	
	\begin{definition}[Exponential tilting operator]
		Set
		\begin{align}\label{Tdef}
			T=&\Gamma(q_\mu) = \exp(\mu \dG (f(\tfrac{s}{d})) ).
		\end{align}
	\end{definition}
	Writing this out explicitly,
	\begin{align}\label{Tdef'}
		T=\exp\left(\mu \sum_{x\in\Lambda}f\left(\tfrac{s(x)}{d}\right)  n_x\right).
	\end{align}
	Observe that $T$ is self-adjoint and invertible; see \eqref{qDef}.\\

	We recall the setup of Theorem \ref{thm:MVB}. In particular, we fix the total particle number to be $N$ and we fix two numbers $0\le\alpha<\beta\le1$. The central idea in the exponential tilting method is to simply write 	
	\[
	\begin{aligned}
		&P_{ \bar N_X\ge \beta}e^{-\mathrm{i}Ht} P_{ \bar N_Y \ge 1-\al}\\
		=& P_{ \bar N_X\ge \beta} T^{-1}   T e^{-\mathrm{i}Ht}T^{-1} T P_{ \bar N_Y \ge 1-\al},
	\end{aligned}
	\]
	which leads to the inequality
	\begin{align}\label{eq:a1}
		&\norm{P_{ \bar N_X\ge \beta}e^{-\mathrm{i}Ht} P_{ \bar N_{Y}\ge 1-\alpha}}\notag\\\le&\norm{P_{ \bar N_X\ge \beta} T^{-1} } \norm{T e^{-\mathrm{i}Ht}T^{-1}}\norm{ T P_{ \bar N_Y \ge 1-\al}}.
	\end{align}    
	We will now estimate each term of the right-hand-side separately.

	\subsection*{Bound on $P_{\bar N_X\ge \beta} T^{-1}$ and $T P_{ \bar N_Y \ge 1-\al}$}
	In this section, we prove the following two bounds.
	\begin{lemma}
		On the $N$-particle sector, we have
		\begin{align}
			\norm{P_{ \bar N_X\ge \beta} T^{-1} } \le& e^{\mu(1-2\beta)N},\label{eq:a2} \\
			\norm{T P_{ \bar N_Y \ge 1-\al}}\le& e^{\mu(2\alpha-1)N}. \label{eq:a3}
		\end{align}
	\end{lemma}
	\begin{proof}
		By Definition \eqref{fDef}, the function  $ f(\tfrac{s(x)}{d})$ is a regularized version of the map $x\mapsto \frac2d s(x)$ in the sense that it   coincides with $x\mapsto \frac2d s(x)$ on $W^0$ and continuously  becomes constant on $X^\infty\cup Y^\infty=(W^0)^\Cp $. Explicitly, 
		\[
		f\left(\frac{s(x)}{d}\right)=\one_{X^\infty}(x)-\one_{Y^\infty}(x)+\frac2d s(x)\one_{W^0}(x).
		\]
		Hence, by the Definition  \eqref{Tdef} of $T$,
		\begin{align*}
			T=&
			\mathrm{exp}\Big(\mu(N_{X^\infty}-N_{Y^\infty})+\frac{2\mu}{d}\mathrm{d}\Gamma(s(x)\one_{W^0}(x))\Big).
		\end{align*}

		We now aim to prove the first estimate \eqref{eq:a2}. 
		Using  that  $-s(x)\le\frac d2$ for $x\in W^0$, we find 
		\[
		-\frac{2\mu}{d}\mathrm{d}\Gamma(s(x)\one_{W^0}(x))
		= -\frac{2\mu}{d}\sum_{x\in W^0}s(x) n_x
		\le \mu N_{W^0}.
		\]
		As both sides of this operator inequality commute (both operators are diagonal in the occupation basis), this implies
		\begin{align*}
			T^{-1}=& \exp\del{\mu(N_{Y^\infty}-N_{X^\infty})-\frac{2\mu}{d}\mathrm{d}\Gamma(s(x)\one_{W^0}(x))}\\
			\le& e^{\mu(N_{Y^\infty\cup W^0}-N_{X^\infty})}.
		\end{align*}
		Recall from Lemma \ref{lemma:XYgeometry} that $X\subset X^\infty$ and $Y\subset Y^\infty$. Hence, on the subspace $\mathrm{Ran}(P_{ \bar N_X\ge \beta})$, we have  $N_{X^\infty}\ge N_X\ge\beta N$ and $N_{Y^\infty\cup W^0}\le N_{X^\Cp }\le(1-\beta)N$, where	$S^\Cp=\Lambda\setminus S$ denotes the complement of $S$ in $\Lam$.
		Combining these estimates, 
		we deduce that 
		\[
		\begin{aligned}
			\norm{T^{-1}P_{ \bar N_X\ge \beta}}
			\le& \norm{e^{\mu(N_{Y^\infty\cup W^0}-N_{X^\infty})}P_{ \bar N_X\ge \beta}}\\
			\le& e^{\mu(1-2\beta)N}.
		\end{aligned}
		\]
		Since $P_{ \bar N_X\ge \beta}$ and $T^{-1}$ are self-adjoint, we have $\norm{P_{ \bar N_X\ge \beta} T^{-1} }=\norm{T^{-1}P_{ \bar N_X\ge \beta}}$ and so \eqref{eq:a2} follows.
		
		The second estimate \eqref{eq:a3} is proven in the same way, using that $N_{Y^\infty}\ge N_Y\ge(1-\alpha)N$ and $N_{X^\infty\cup W^0}\le N_{Y^\Cp }\le\alpha N$ on $\mathrm{Ran}(P_{ \bar N_Y \ge 1-\al})$, together with $s(x)\le\frac d2$ for $x\in W^0$. 
	\end{proof}

	\subsection*{Bound on the deformed propagator}
	To bound \eqref{eq:a1}, it remains to estimate the norm of the deformed propagator $T e^{-\mathrm{i}Ht} T^{-1}$. 
	\begin{lemma}\label{lem8}
		Suppose that Assumptions \ref{C1}--\ref{C2} on the Hamiltonian hold. Then, on the $N$-particle sector, we have, for all $t\in\mathbb{R}$,
		\begin{equation}\label{eq:a6}
			\norm{T e^{-\mathrm{i}Ht} T^{-1}}\le e^{aNv |t|}.
		\end{equation}
	\end{lemma}
	\begin{proof}
		For any bounded operator $A$, we abbreviate 	 $\tilde A=T AT^{-1}$. 	Since $V(\{n_x\})$ commutes with $T$, we have $\tilde H=\tilde H_{0}+V$ and so 
		\begin{equation*}
			\tilde U_{t}=	T e^{-\mathrm{i}Ht} T^{-1}=e^{-\mathrm{i}t\tilde H}=e^{-\mathrm{i}t(\tilde H_{0}+V)}.
		\end{equation*}
		
		For a bounded operator $A$, we also denote $\im A=\frac{A-A^\dagger}{2\mathrm{i}}$, which is always self-adjoint.  Given any state $\psi$ in the $N$-particle sector, we compute
		\begin{align}
			\label{S21}
			\di_t\norm{\tilde U_{t}\psi}^2=&2\br{\tilde U_{t}\psi, (\im \tilde H_{0})\tilde U_{t}\psi }\notag\\
			\le& 2\sup\mathrm{spec}\, (\im \tilde H_{0})\norm{\tilde U_{t}\psi}^2.
		\end{align}
		Here, for a self-adjoint operator $A$, $\sup\mathrm{spec}(A)$ refers to the supremum over the spectrum of $A$.

		Using Gronwall's lemma and taking the supremum over normalized $N$-particle states $\psi$, it follows that
		\[
		\|\tilde U_{t}\|
		\le \begin{cases}
			e^{t\sup\mathrm{spec}\, (\im \tilde H_{0})},\quad &t>0,\\
			e^{-t\inf\mathrm{spec}\, (\im \tilde H_{0})},\quad &t<0,
		\end{cases}
		\]
		and so
		\begin{equation}\label{eq:a4}
			\|\tilde U_{t}\|\leq e^{|t| \|\im \tilde H_{0}\|}.
		\end{equation}
		
		It thus remains to bound $\|\im \tilde H_{0}\|$.  We first calculate $\im \tilde H_{0}$ by using Proposition \ref{proposition:Gammaproperties} (ii) with $q=q_\mu$ from \eqref{qDef}. This gives
		\begin{align*}
			\tilde H_{0}
			=&T H_0 T^{-1}\\
			=&  \Gamma(q_\mu) \left(\sum_{x,y \in \Lambda} J_{xy} a_x^\dagger a_y\right) \Gamma(q_\mu^{-1})\notag\\
			=& \sum_{x,y \in \Lambda} J_{xy} q_\mu(x)q_\mu^{-1}(y) a_x^\dagger a_y
			\notag\\=& \sum_{x,y \in \Lambda} J_{xy}\exp(\mu(f(\tfrac{s(x)}{d}) - f(\tfrac{s(y)}{d} )) ) a_x^\dagger a_y.
		\end{align*}
		
		Since $J_{xy}=\bar J_{yx}$, we find
		\begin{align*}
			\im 	\tilde H_{0}=&  \dG(\tilde J),\\
			\textnormal{for }\tilde J_{xy}=&\frac1{\mathrm{i}}J_{xy}\sinh(\mu (f(\tfrac{s(x)}{d}) - f(\tfrac{s(y)}{d} ))).
		\end{align*}
		By Proposition \ref{proposition:Gammaproperties} (i), we have
		\begin{equation}\label{216}
			\|\im 	\tilde H_{0}\|=\|\dG(\tilde J)\|\leq N\|\tilde J\|
		\end{equation}
		and so it remains to bound the norm of the deformed hopping matrix, $\|\tilde J\|$.
		

		To this end, observe that $f(s)$ satisfies $\abs{f(s)-f(s')}\le \min (2,2 \abs{s-s'})$ for all $s,s'\in\Rb$ (see \eqref{fDef}). Using this and the fact   that    $s(x)$ is $1$-Lipschitz, cf.\ \eqref{sLip}, we have  
		\begin{align*}
			\abs{f(\tfrac{s(x)}{d}) - f(\tfrac{s(y)}{d} )}
			\le \min (2,\tfrac{2}{d}\abs{x-y}),\quad x,\,y\in\Lam. 
		\end{align*}
		Recalling that $\mu=\tfrac{da}{2}$, this  implies
		\begin{align}
			\label{qId}
			\abs{\tilde J_{xy}}\le \abs{J_{xy}}\sinh(a\min\{d,|x-y|\}).
		\end{align}
		
		By the Schur test for matrix norms and the short-range Assumption \ref{C1},
		\[
		\|\tilde J\|\leq \max_{x\in\Lambda}\sum_{y\in\Lambda}|J_{xy}|
		\sinh(a|x-y|) \le av.
		\]
		Combining this estimate with \eqref{eq:a4} and \eqref{216} proves the lemma.
	\end{proof}

	We now have all the ingredients in place to prove our main result.
	
	\begin{proof}[Proof of Theorem \ref{thm:MVB}]
		Combining \eqref{eq:a1}, \eqref{eq:a2}, \eqref{eq:a3} and \eqref{eq:a6} and recalling that $\mu=\frac{a}{2}d_{XY}$, we find that
		\[
		\begin{aligned}
			&\norm{P_{ \bar N_X\ge \beta}U_t P_{ \bar N_Y \ge 1-\al}} \\
			\le& \exp(2\mu(\beta-\alpha)N)
			\exp( -|t|v aN)\\
			=& \mathrm{exp}\big(-aN\big[(\beta-\alpha)d_{XY}-v|t|\big]\big),
		\end{aligned}
		\]
		which proves \eqref{MVB}.
	\end{proof}
	
	We remark that for time-dependent Hamiltonian $H(t)$ satisfying Assumptions \ref{C1} and \ref{C2} uniformly for all times, inequalities \eqref{eq:a4} -- \eqref{216} remain valid, and therefore \lemref{lem8} generalizes to $H(t)$, upon replacing $e^{-\mathrm{i}Ht}$ by the usual time-ordered propagator
	\[
	U(t,0)=\mathcal T \exp\left(\int_0^t H(s) \mathrm{d} s\right).
	\] Since \lemref{lem8} is the only place where propagator estimate is involved (see \eqref{eq:a1}), the conclusion of \thmref{thm:MVB} extends to $H(t)$, with $e^{-\mathrm{i}Ht}$ replaced by $U(t,0)$ in eq.~\eqref{MVB}.
\end{document}